\documentclass[letterpaper, 10 pt, conference]{ieeeconf}  

\IEEEoverridecommandlockouts 
\newtheorem{theorem}{Theorem}
\newtheorem{proposition}{Proposition}

\usepackage{packagearticleieee}
\usepackage{packagevariables} % I define all variables here
\usepackage{algorithm} 
\usepackage{algpseudocode}
\usepackage[utf8]{inputenc}
\makeatletter
\let\NAT@parse\undefined
\makeatother
\usepackage[hyphens]{url}
\usepackage[colorlinks=true,linkcolor=blue!80!black,citecolor=blue!80!black,urlcolor=blue!80!black]{hyperref}
\newcommand{\Dj}{(D_-)_j}

\title{\LARGE\bf Robust Data-Driven Tube-Based Zonotopic Predictive Control \\with Closed-Loop Guarantees}

\author{
Mahsa Farjadnia$^{1,\star}$, Angela Fontan$^2$, Amr Alanwar$^3$, Marco Molinari$^{1}$, and Karl Henrik Johansson$^2$%
\thanks{
This work was supported by the Swedish Energy Authority and IQ Samhällsbyggnad project DOCENT, project number P2023-01513, agreement 2023-205321, by SSF, grant RIT17-0046, by the Digital Futures project HiSS, by the Swedish Research Council Distinguished Professor Grant 2017-01078, and by the Knut and Alice Wallenberg Foundation Wallenberg Scholar grant.
}
\thanks{$^\star$Corresponding author.}
\thanks{$^1$%
M. Farjadnia and M. Molinari are with the
Department of Energy Technology, 
KTH Royal Institute of Technology. E-mail: \{mahsafa,marcomo\}@kth.se.}%
\thanks{$^2$%
A. Fontan and K. H. Johansson are with the Division of Decision and Control Systems, KTH Royal Institute of Technology. E-mail: \{angfon,kallej\}@kth.se.}%
\thanks{$^3$A. Alanwar is with the School of Computation, Information and Technology, Technical University of Munich. E-mail:  alanwar@tum.de.}
\thanks{$^{1,2}$The authors are also affiliated with Digital Futures.}
}

\begin{document}\bstctlcite{IEEEexample:BSTcontrol}

\newcolumntype{C}{>{$}l<{$}}
\maketitle

\begin{abstract}
This work proposes a robust data-driven tube-based zonotopic predictive control (TZPC) approach for discrete-time linear systems, designed to ensure stability and recursive feasibility in the presence of bounded noise. The proposed approach consists of two phases. In an initial learning phase, we provide an over-approximation of all models consistent with past input and noisy state data using zonotope properties. Subsequently, in a control phase, we formulate an optimization problem, which by integrating terminal ingredients is proven to be recursively feasible. Moreover, we prove that implementing this data-driven predictive control approach guarantees robust exponential stability of the closed-loop system. The effectiveness and competitive performance of the proposed control strategy, compared to recent data-driven predictive control methods, are illustrated through numerical simulations.
\end{abstract}
\section{Introduction}
Model predictive control (MPC) is one of the well-established techniques in the field of control systems, renowned for its proficiency in handling system constraints, nonlinear dynamics, and trajectory tracking \cite{rawlings2017model}. Knowledge of the dynamical system model is essential for MPC to guarantee optimal control performance and satisfaction of system constraints. However, achieving this in practice can be challenging, especially for complex cyber-physical-human systems. One such application where these challenges become particularly evident is in building automation \cite{8745685}.

The building sector accounts for approximately 40\% of the total final energy supply \cite{IEA2020}, with 76\% of that consumed by heating, ventilation, and air conditioning (HVAC) systems for comfort control \cite{perez2008review}. Enhancing the sustainability and efficiency of the built environment is imperative and different efforts and initiatives are already in place \cite{Molinari2023Bottlenecks,Farjadnia2023Windows,Fontan2023Socialabbrv,behrunani2023degradation}.

Data-driven methodologies for predictive control are increasingly recognized within the control engineering community to address these challenges \cite{hou2013model}.
The pioneering work \cite{coulson2019data} represents a significant contribution to this field, introducing the data-enabled predictive control algorithm (DeePC) to address the optimal trajectory tracking problem for unknown linear time-invariant (LTI) systems. Originally based on behavioral systems theory, DeePC was limited to noise-free data scenarios, constraining its applicability in real-world settings. To address this limitation, the authors of \cite{coulson2021distributionally} extended DeePC to consider stochastic LTI systems.
A rich literature on data-driven MPC is now present, addressing both noise-free~\cite{berberich2020data} and noisy conditions~\cite{Alanwar2022Robust,Berberich2021Guarantees}. Recent studies have been exploiting zonotope properties to handle unknown but bounded system noises, in both linear \cite{Alanwar2022Robust,Alanwar2023Datadriven,Russo2023TZDDPC} and nonlinear systems \cite{Farjadnia2023NZPC}. 
In \cite{Alanwar2022Robust}, the proposed robust data-driven zonotopic-based (ZPC) approach employed data-driven predictions of reachable sets within a predictive control framework, providing robust constraint satisfaction.
Despite its novel contribution, this work fell short of demonstrating recursive feasibility or offering any closed-loop stability guarantees.
Building on \cite{Alanwar2022Robust}, the work \cite{Russo2023TZDDPC} provided closed-loop stability via a tube-based zonotopic data-driven predictive control (TZDDPC) approach, although it similarly lacked the analysis of recursive feasibility.

The aim of this work is to develop and analyze a robust data-driven tube-based zonotopic predictive control (TZPC) method for unknown LTI systems, subject to bounded noise as well as state and input constraints. This work builds on the results presented in \cite{Alanwar2022Robust,Alanwar2023Datadriven}. However, in contrast to the approach taken in \cite{Alanwar2023Datadriven} (and accordingly \cite{Russo2023TZDDPC}), TZPC introduces two critical contributions. First, it proves the recursive feasibility of the optimal control problem. Second, it guarantees robust exponential stability for the closed-loop system, which is achieved by exploiting terminal ingredients in the TZPC formulation. Finally, this paper presents a comparative analysis, evaluating both the performance and computational complexity of the proposed TZPC strategy relative to the ZPC method \cite{Alanwar2022Robust} and to the TZDDPC method~\cite{Russo2023TZDDPC}. In particular, it is shown that employing reachability analysis concepts implicitly within TZPC significantly reduces execution time, as opposed to the direct incorporation of reachable sets in ZPC and TZDDPC. 

The remainder of the paper is organized as follows. Section~\ref{sec:preliminaries} introduces technical preliminaries and the notation used in the paper. Section~\ref{sec:problem-formulation} provides the problem formulation. Section~\ref{sec:TZPC} presents the novel TZPC approach and includes a design procedure for the terminal ingredients. Section \ref{sec:guarantees} discusses the closed-loop guarantees and Section~\ref{sec:numerical-simulations} illustrates two numerical simulations. Finally, Section~\ref{sec:conclusion} offers conclusive remarks.

\section{Technical preliminaries}\label{sec:preliminaries}

\subsection{Notations}
The set of natural numbers, nonnegative integers, real $n$-dimensional space, and real $n \times m$ matrix space are denoted by $\N$, $\Z_{\geq 0}$, $\R^{n}$, and $\R^{n \times m}$. A positive (resp. negative) definite matrix $A$ is denoted by $A>0$ (resp. $A<0$). For a matrix $A$, $\norm{A}_F$ is its Frobenius norm.
For a vector $x$, $\abs{x}$ is its element-wise absolute value, $\Vert x \Vert$ is its 2-norm, and $\Vert x \Vert_P^2=x^{\top}P x$, where $P>0$. The distance of a point $x$ from a set $\mathcal{S}$ is defined as $\norm{x}_\mathcal{S} = \text{inf}\{\Vert x-z\Vert \big| z \in \mathcal{S} \}$. 
%For a given vector $x$, $x_i$ denotes its $i$th element. 
For a matrix $A$, $(A)_j$ denotes its $j$th column, $A^{\top}$ its transpose, and $A^\dagger$ its Moore-Penrose pseudoinverse. $\diag \cdot$ is the diagonal operator making a diagonal matrix by its arguments.
$\1_n$ and $\zero_n$ are the n-dimensional vectors with all entries equal to $1$ and $0$, respectively. Where no confusion arises, we use $\1$ and $\zero$ to denote matrices with the proper dimensions.

\subsection{Set representation}
A zonotope is an affine transformation of a unit hypercube, and it is defined as follows \cite{Kuhn1998}.
\begin{definition}[Zonotope]
    Given a center $c_{\zon} \in \R^n $ and a number $\gamma_\zon\in \N$ of generator vectors $g_{\zon_i}\in\mathbb{R}^n$, a zonotope is defined as
    \begin{equation*}
        \zon = \Big\{ x \in \R^{n} \Big| x = c_{\zon} + \sum_{i=1}^{\gamma_\zon} \beta_i\, g_{\zon_i},-1 \leq \beta_i\leq 1 \Big\}.
    \end{equation*}
    We use the short notation $ \zon = \langle c_{\zon},G_{\zon} \rangle$ to denote a zonotope, where $G_{\zon} = [g_{\zon_1}\dots g_{\zon_{\gamma_\zon}}]\in \R^{n\times \gamma_\zon }$. %\hfill $\lrcorner$
\end{definition}
A zonotope $\zon$ can be over-approximated by a multidimensional interval as \cite[Proposition~2.2]{Althoff2010PhD}: 
\begin{equation*}
I = \Int \zon = [c_{\zon}-\Delta g,c_{\zon}+\Delta g], \;\; \Delta g=\sum_{i=1}^{\gamma_\zon} |g_{\zon_i}|.
\end{equation*}
The conversion of an interval $I=[\underline{I},\overline{I}]$ to a zonotope is denoted by $\zon =  \mathrm{zonotope}(\underline{I},\overline{I})$.

Given $L\in\R^{m\times n}$, the linear transformation of a zonotope $\zon$ is a zonotope $L\zon=\langle Lc_{\zon},LG_{\zon} \rangle$. Given two zonotopes $\zon_1 = \langle c_{\zon_1 },G_{\zon_1 } \rangle \subset \R^n$ and $\zon_2 = \langle c_{\zon_2 },G_{\zon_2 } \rangle\subset \R^n$, the Minkowski sum, the Minkowski difference, and the Cartesian product are respectively computed as:
\allowdisplaybreaks
\begin{align*}
\zon_1 \oplus \zon_2  
    &= \Big\langle c_{\zon_1 }+ c_{\zon_1 },[G_{\zon_1 } \, \, G_{\zon_2 }] \Big\rangle,
    \\
    \zon_1\ominus \zon_2  
    &= \{ z_1 \in \R^n\,|\, z_1 \oplus \zon_2 \subseteq \zon_1\},
\\
    \mathcal{Z}_1 \times \zon_2 
    &= \left\langle \begin{bmatrix} c_{\zon_1}\\ 
    c_{\mathcal{Z}_2} \end{bmatrix},
    \begin{bmatrix} G_{\zon_1} & \zero \\ \zero & G_{\zon_2} \end{bmatrix}
   \right\rangle.
\end{align*}

A matrix zonotope is presented as $ \mathcal{M} = \langle C_{\mathcal{M}}, G_{\mathcal{M}} \rangle$, and its definition is analogous to that of a zonotope.  It is characterized by a center matrix $C_{\mzon} \in \mathbb{R}^{n\times m}$  and $\gamma_{\mzon} \in \N$ generator matrices $G_{\mathcal{M}} = \begin{bmatrix}G_{\mzon_1}& \cdots & G_{\mzon_{\gamma_\mzon}} \end{bmatrix}\in \R^{n\times (m\gamma_{\mzon}) }$ \cite[p.~52]{Althoff2010PhD}. A matrix zonotope can be over-approximated by an interval matrix $I_\mathcal{M} =\Int{\mathcal{M}} =  [\underline{\mathcal{M}},\overline{\mathcal{M}}]$, where $\underline{\mathcal{M}} \in \R^{n\times m}$ and $\overline{\mathcal{M}} \in \R^{n\times m}$ denote its lower bound and upper bound, respectively. The Frobenius norm of the interval matrix $I_\mathcal{M}$ is defined as
$\Fnorm{I_\mathcal{M}}_F = \norm{|I_{\mathcal{M}_c}|+\Delta}_F$ with $I_{\mathcal{M}_c} = 0.5(\overline{\mathcal{M}}+\underline{\mathcal{M}})$ and $\Delta = 0.5(\overline{\mathcal{M}}-\underline{\mathcal{M}})$ \cite{farhadsefat2011norms}.

\subsection{Set invariance and stability}
This section introduces the concepts of set invariance and robust exponential stability for systems in the presence of bounded noise \cite{mayne2005robust}.

Consider a system with state constraint set $\mathcal{X}$ and noise bounded by a zonotope $\zon_w$,
\begin{align}
    x(k+&1)=f(x(k),w(k)),  \notag \\ &x(k)\in \mathcal{X},  \; w(k)\in \zon_w,\forall\, k\in \mathbb{Z}_{\geq 0}.
    \label{eq:sys-preliminaries}
\end{align}

\begin{definition}[Robustly positive invariant set]
A set $\mathcal{S} \subset \R^{n_x}$ is robustly positive invariant for the system \eqref{eq:sys-preliminaries} if $\mathcal{S} \subseteq \mathcal{X}$ and $x(k+1) \in \mathcal{S}$, $\forall\, x(k) \in \mathcal{S}$, $\forall\, w(k) \in \zon_w,\forall\, k\in \mathbb{Z}_{\geq 0} $.
\end{definition}
% \comment{Add here the definition of persistently exciting input signal (see e.g. \cite{Alanwar2022Robust})}.
\begin{definition}[Robust exponential stability]
\label{def:robust-exp-stability}
A set $\mathcal{S}$ is robustly exponentially stable for the system~\eqref{eq:sys-preliminaries}, with a region of attraction $\Omega$, if there exist a $c > 0$ and a $\gamma \in (0,1)$ such that any solution $x(k)$ of~\eqref{eq:sys-preliminaries} with initial state $x(0)\in\Omega$ satisfies $\norm {x(k)}_\mathcal{S} \le c \gamma^k\norm {x(0)}_\mathcal{S}$, $\forall\, k\in \mathbb{Z}_{\geq 0}$.
\end{definition}

\section{Problem formulation}\label{sec:problem-formulation}
We consider the following discrete-time linear control system 
\begin{equation}
    x(k+1) = A x(k) + B u(k) + w(k),
\label{eq:model-linear}
\end{equation}
where $A\in \R^{n_x \times n_x}$ and $B\in \R^{n_x\times n_u}$ are unknown system dynamics matrices, $x(k) \in \zon_x =\langle c_{\zon_x},G_{\zon_x} \rangle \subset  \R^{n_x}$ and $u(k) \in \zon_u = \langle c_{\zon_u},G_{\zon_u} \rangle \subset \R^{n_u}$ are respectively the state and the control input of the system at time $k \in \mathbb{Z}_{\geq 0}$, and $w(k)$ denotes the noise. 
%process noise or the system disturbance
%At each time step $k$, 
Input and state constraints for a controller design are 
\begin{align}\label{eq:constraints}
    u(k) \in \mathcal{U}\subseteq \zon_u , \quad x(k) \in \mathcal{X}\subseteq \zon_x .
 \end{align}
In \eqref{eq:constraints}, $\zon_u$ and $\zon_x$ represent the time-invariant domains of control inputs and states, respectively, corresponding to inherent constraints of the problem and possibly driven by physical limitations. %Meanwhile, the sets $\mathcal{U}$ and $\mathcal{X}$ are input and state constraints for the controller design. 
We make the following standing assumptions, which will hold throughout the paper.
%on \eqref{eq:model-linear}.
\begin{assumption}\label{ass:controllability}
The pair $(A, B)$ is controllable.
\end{assumption}
\begin{assumption}\label{ass:zon-noise}
The noise $w$ is bounded by a known zonotope, i.e., $w(k)\in \zon_w = \langle c_{\zon_w},G_{\zon_w} \rangle$ $\forall\, k\in \mathbb{Z}_{\geq 0}$, which includes the origin. %and it
\end{assumption}

We aim to solve a receding horizon optimal control problem, ensuring recursive feasibility and robust exponential stability as closed-loop guarantees, by employing only the input and noisy state data of the system \eqref{eq:model-linear} and without explicit knowledge of the system matrices. To be precise, we assume access to past $n_T$ input-state trajectories of different lengths, denoted by $T_i+1$, for $i=1,\dots,n_T$. These trajectories are denoted as $\{ u(k)\}^{0}_{k=-T_i}$ and $\{ x(k)\}^{0}_{k=-T_i}$, where a negative data point index $k$ refers to data obtained from offline experiments. For notational simplicity, in the theoretical formulation, we consider a single trajectory ($n_T=1$) of adequate length $T$ and collect all the input and noisy state data into the following matrices:
\begin{align*}
     X_{+} &= \begin{bmatrix}  x(-T+1) & x(-T) &\cdots & x(0)\end{bmatrix}, \\
     X_{-} &= \begin{bmatrix}  x(-T)& x(-T+1) & \cdots & x(-1)\end{bmatrix}, \\
      U_{-} &= \begin{bmatrix} u(-T) &  u(-T+1)  & \cdots & u(-1) \end{bmatrix}.
\end{align*} 
We define $D_-= \begin{bmatrix} X_{-}^\top&U_{-}^\top\end{bmatrix}^\top$ and denote all available past data by $D= \begin{bmatrix} X_{+}^\top&X_{-}^\top &U_{-}^\top\end{bmatrix}^\top$. We consider the following standing assumption, necessary for the data-driven approach explained in the next section.
\begin{assumption}\label{ass:rank_D}
    We assume that the data matrix $D_-$ has full row rank, i.e., $\mathrm{rank}(D_-) = n_x + n_u$.
\end{assumption}

We denote the sequence of the unknown noise corresponding to the available input-state trajectories by $\{w(k)\}_{k=-T}^{0}$. Based on Assumption~\ref{ass:zon-noise}, we can directly infer that $W_{-} = \begin{bmatrix}  w(-T)& w(-T+1) & \cdots & w(-1)\end{bmatrix} \in\mzon_w = \langle C_{\mzon_w},G_{\mzon_w} \rangle$ with $C_{\mzon_w} \in \R^{n_x \times n_T}$ and $G_{\mzon_w} \in \R ^{n_x \times \gamma_\zon n_T}$. Here, $\mzon_w$ represents the matrix zonotope resulting from the concatenation of multiple noise zonotopes $\zon_w$ \cite{Alanwar2023Datadriven}.

\section{Robust Data-driven Tube-based Predictive Control}\label{sec:TZPC}
This section describes the novel TZPC approach for the unknown linear system~\eqref{eq:model-linear} corrupted by noise. 
The proposed TZPC algorithm consists of two phases: an offline learning phase (Section~\ref{subsec:Learning}) and an online control phase (Section~\ref{subsec:Control}).

\subsection{Learning Phase} \label{subsec:Learning} 
 The goal of the learning phase is to replace the model-based description in traditional MPC with a data-driven system representation. This takes the form of a matrix zonotope, serving as an over-approximation of all models consistent with the available input and noisy-state data (Lemma~\ref{lem:setofAB}). 
We will then utilize concepts from the tube-based robust MPC literature to define a nominal system \eqref{eq:nominal-model}, a feedback policy \eqref{eq:feedback-policy}, and an associated error dynamics \eqref{eq:controlled-sys-RPI}.

\begin{lemma}[Lemma~1,~\cite{Alanwar2023Datadriven}]\label{lem:setofAB}
Given the input-state trajectories $D$ of the system~\eqref{eq:model-linear}, the matrix zonotope
\begin{align}
    \mzon_D = (X_{+} \oplus- \mzon_w) D_-^\dagger,
   \label{eq:zonoAB}
\end{align}  
with $D_-=\begin{bmatrix} X_-^\top & U_-^\top \end{bmatrix}^\top$, contains all system dynamics matrices $\begin{bmatrix}A & B \end{bmatrix}$ that are consistent with the data $D$ and the noise bound $\mzon_w$. 
\end{lemma}
The offline data-collection phase, referred to as the learning phase to be consistent with \cite{Farjadnia2023NZPC}, as described in Lemma~\ref{lem:setofAB}, necessitates the existence of a right pseudoinverse for the matrix $D_-$, ensured by Assumption~\ref{ass:rank_D}.
 
As in standard tube-based MPC, we define a nominal system for the model~\eqref{eq:model-linear} as follows:
\begin{equation}\label{eq:nominal-model}
    \bar{x}(k+1) = \bar{A}\bar{x}(k)+\bar{B} \bar{u}(k),
\end{equation}
where $\bar{x}$ and $\bar{u}$ are nominal state and control input, and the matrix $\bar{M} =\begin{bmatrix} \bar{A} & \bar{B} \end{bmatrix}$ is chosen from $\mzon_D$. The following assumption ensures the existence of a local stabilizing controller, common to all models in $\mzon_D$ consistent with the data. This in turn implies the stabilizability of the selected pair $(\bar{A},\bar{B})$. 

\begin{assumption}
\label{ass:localcontroller}
There exists a matrix $K\in \R^{n_x\times n_u}$ and a positive definite matrix $P=P^\top\in \R^{n_x \times n_x}$ such that, for any chosen nominal model $\begin{bmatrix} \bar{A} & \bar{B} \end{bmatrix} \in \mzon_D$, and a positive definite matrix $Q_K = Q_K^\top\in \R^{n_x \times n_x}$, the inequality $\bar{A}_K^T P \bar{A}_K - P <-Q_K$, with $\bar{A}_K = \bar{A}+\bar{B}K$, holds. 
\end{assumption}

Considering the nominal model~\eqref{eq:nominal-model}, the feedback policy for the system~\eqref{eq:model-linear} is defined as 
\begin{equation}
    u(k) = \bar{u}(k) + K (x(k)-\bar{x}(k)).
\label{eq:feedback-policy}
\end{equation}
We can rewrite the system dynamics as follows:
\begin{equation}
    x(k+1) = (\bar{M}+\Delta M) \begin{bmatrix} x(k)\\u(k) \end{bmatrix}+w(k),
    \label{eq:model-mismatch}
\end{equation}
where $\Delta M = M - \bar{M}$ is the model mismatch between the true unknown system matrix $M = \begin{bmatrix} A & B \end{bmatrix}$ and the nominal system matrix $\bar{M}$. %with $M = \begin{bmatrix} A & B \end{bmatrix}$.
\begin{lemma}\label{lem:mismatch}
Given input-state trajectories $D$ of \eqref{eq:model-linear} and $\zon_w$, the model mismatch $\Delta M$ %in the system~\eqref{eq:model-mismatch} 
is over-approximated for all  $(x,u) \in \zon_x \times \zon_u$ as:
\begin{equation}
    \Delta M \begin{bmatrix} x(k)\\u(k) 
    \end{bmatrix}\in \zon_M \oplus \zon_\epsilon,
    \label{eq:model-mismatch-bound}
\end{equation}
where the zonotopes $\zon_M$ and $\zon_\epsilon$ are computed as
\begin{equation}\label{eq:Z_M}
    \zon_M= \mathrm{zonotope}(\underline{z}_m,\overline{z}_m)\oplus- \zon_w,
\end{equation}
% $\zon_M= \mathrm{zonotope}(\underline{z}_m,\overline{z}_m)- \zon_w$ 
with $\underline{z}_m  = \min_{j=1,\dots,T} \big( {(X_{+})}_{j} - \bar{M} \Dj \big),$ $
      \overline{z}_m  = \max_{j=1,\dots,T} \big( {(X_{+})}_{j} - \bar{M}\Dj\big),$
and 
$$\zon_\epsilon = \big\langle \zero_{n_x},\mathrm{diag}(\Fnorm{I_{\mzon_D}}_F\, \delta/2,\dots, \Fnorm{I_{\mzon_D}}_F\,\delta/2)\big\rangle,$$ where $I_{\mzon_D} = \Int {\mzon_D}$ and
$\delta$ is the covering radius of $D_-$.
\end{lemma}
\begin{proof}
First observe that the equality $(X_+)_{j} - (W_-)_{j} =  (\bar{M}+\Delta M)\,\Dj $ holds for any $\Dj$ and its corresponding noise realization $(W_{-})_{j} \in \zon_w$. This implies that
\begin{equation}
    \Delta M \, \Dj \in  (X_+)_{j} -\bar{M} \,\Dj \oplus- \zon_w . 
\label{eq: pointmodelerrorbound}
 \end{equation}
By over-approximating the right-hand side of \eqref{eq: pointmodelerrorbound} for all $\Dj$, $j = 1, \dots, T$, together with their noise realizations, we obtain $\zon_M$ as outlined in~\eqref{eq:Z_M}.

To complete the proof and show \eqref{eq:model-mismatch-bound}, we over-approximate the model mismatch for all $(x,u)\in \zon_x\times \zon_u$. Since $\zon_x {\times} \zon_{u}$ is compact, it follows that all points $ \Dj$, $j=1,\dots,T$, are dense in $\zon_x {\times} \zon_{u}$. In other words, there exists some $\delta \geq 0$ such that, for any $[x(k)^\top \, u(k)^\top]^{\top} \in \zon_x \times \zon_{u}$, there exists a $\Dj$ for which $\norm{[x(k)^\top \, u(k)^\top]^{\top}-\Dj } \le \delta$ holds (details in \cite{montenbruck2016some}). 
Here, $\delta$ represents the covering radius for $D_-$. Then, for any $[x(k)^\top \, u(k)^\top]^{\top}\in \zon_x \times \zon_{u}$ and $\Dj$, it holds that
\begin{align*}
&\norm{\begin{bmatrix} A & B \end{bmatrix} ([x(k)^\top \, u(k)^\top]^{\top}-\Dj)}
\\ 
&\le \norm{\begin{bmatrix} A & B 
\end{bmatrix}} \norm{ [x(k)^\top \, u(k)^\top]^{\top}\!-\Dj} \le \Fnorm{I_{\mzon_D}}_F\delta,
\end{align*}
which implies \eqref{eq:model-mismatch-bound}.  
\end{proof}
\begin{remark}[Computing $\delta$] 
To estimate the covering radius $\delta$ (see \cite{montenbruck2016some} for the definition of covering radius), we follow the method presented in \cite[Remark 11]{Alanwar2023Datadriven} using the available data $D$.
\end{remark}

Before discussing the control phase in the next section, we define the error between the measured state and the nominal state as $e(k)=x(k)-\bar{x}(k)$. The controlled error dynamic is then given by
\begin{equation}\label{eq:controlled-sys-RPI}
    e(k+1) = \bar{A}_K e(k)+ \phi(k),
\end{equation}
with    
\begin{equation*}
\phi(k) = \Delta M \begin{bmatrix} x(k)\\u(k) \end{bmatrix}+w(k) \in \zon_\phi= \zon_M \oplus \zon_\epsilon \oplus \zon_w,
\end{equation*} 
where $\zon_M$ and $ \zon_\epsilon$ are defined in Lemma~\ref{lem:mismatch}.
Without loss of generality, it is assumed that the set $\zon_M$ encompasses the origin, which, in turn, implies that the set $\zon_\phi$ also contains the origin. Under Assumption~\ref{ass:localcontroller}, we know that there exists a robustly positive invariant set $\mathcal{S}$ for the error dynamics~\eqref{eq:controlled-sys-RPI}, i.e., $\bar{A}_K \mathcal{S} \oplus \zon_\phi\subseteq \mathcal{S}$ \cite[Sec. IV]{kolmanovsky1998theory}. This guarantees the true trajectory $x$ to be in a robustly positive invariant set  $\mathcal{S}$ around the nominal trajectory $\bar{x}$, as we will prove in the following sections. The computation of $\mathcal{S}$ will be discussed in the next section.
    
\subsection{Control phase}\label{subsec:Control}
In this section, we propose a data-driven predictive controller that leverages the data-driven system representation obtained in the learning phase (Section~\ref{subsec:Learning}). This tube-based technique employs a robust feedback controller keeping the noisy system~\eqref{eq:model-linear} in a robustly positive invariant set around a nominal trajectory~\eqref{eq:nominal-model}. To guarantee the satisfaction of the state/input constraints \eqref{eq:constraints}, tighter constraints are designed for the noise-free nominal system \eqref{eq:nominal-model} and thus for the state/input decision variables. We formulate the following data-driven optimal control problem at time $t$:
\begin{subequations}
\begin{align}
\min_{\bar{u}_t,\bar{x}_t} \quad& \!\!\!\!\!\!\!\!\sum_{k=t}^{t+N-1} \!\! \ell(\bar{x}(k|t),\bar{u}(k|t)) {+} \ell_N(\bar{x}(t+N|t)) \label{eq:cost-function}\\
\text{s.t.}\quad&  \bar{x}(k+1|t) = \bar{A}\bar{x}(k|t) +\bar{B} \bar{u}(k|t), \label{eq:nominalconst}\\
&  \bar{u}(k|t) \in \mathcal{U} \ominus K\mathcal{S}, \label{eq:ubar-constr}\\
& \bar{x}(k|t) \in \mathcal{X}\ominus \mathcal{S}, \label{eq:xbar-constr}\\
& x(t) \in \bar{x}(t|t) \oplus \mathcal{S},\label{eq:x-constr}\\
& \bar{x}(t+N|t)\in  \mathcal{X}_N \subseteq \mathcal{X}\ominus \mathcal{S}.\label{eq:terminal-constr}
\end{align}
\label{eq:opt-zonopc}
\end{subequations}
Here, $ \bar{u}_t = \{\bar{u}(t|t), \dots, \bar{u}(t+N-1|t)\}$ and $ \bar{x}_t = \{\bar{x}(t|t), \dots, \bar{x}(t+N|t)\}$ are the sequences of input and state decision variables, respectively, predicted at time $t$ over the finite horizon $N$. Let $[x_s^\top \; u_s^\top]^\top$ denote a desired setpoint of the system~\eqref{eq:model-linear}. The cost function in \eqref{eq:cost-function} contains a stage cost $\ell(\bar{x}(k|t),\bar{u}(k|t)) = \norm{\bar{x}(k|t) -x_s}_Q^2 + \norm{\bar{u}(k|t) - u_s}_R^2$, where $Q>0$ and  $R>0$ are design parameters. Additionally, it contains a terminal cost characterized by the matrix $P>0$, presented as $\ell_N (\bar{x}(t+N|t))= \norm{\bar{x}(t+N|t) -x_s}_P^2$. Without loss of generality, we let $u_s = \zero_{n_u}$ and $x_s = \zero_{n_x}$. Equation \eqref{eq:nominalconst} indicates the nominal system dynamics, while \eqref{eq:ubar-constr} and \eqref{eq:xbar-constr} represent the tightened input and state constraints to be satisfied. These constraints implicitly incorporate reachability analysis in the optimal problem~\eqref{eq:opt-zonopc}, differently from ZPC \cite{Alanwar2022Robust}, where the allowable reachable sets are explicitly computed at each time step. The relationship between the predicted state at current time $\bar{x}(t|t)$ and the measured state $x(t)$ is specified in \eqref{eq:x-constr}. Finally, in \eqref{eq:terminal-constr}, $\mathcal{X}_N$ is the terminal constraint on the predicted state at time $t+N$.

Problem~\eqref{eq:opt-zonopc} is a convex optimization problem, and we denote its optimal solution at time $t$ by $(\bar{u}^*_t,\bar{x}^*_t)$. From \eqref{eq:x-constr}, $\bar{x}^*(t|t)$ can be expressed as $ \bar{x}^*(x(t))$, which we abbreviate as $\bar{x}^*(t)$. Accordingly, the optimal cost in \eqref{eq:cost-function} is denoted by $J^*_\ell(\bar{x}^*(t))$. Similar to standard MPC, the problem~\eqref{eq:opt-zonopc} is solved in a receding horizon fashion (Algorithm~\ref{alg:zonopc}), where we use the first optimal control input $\bar{u}^*(t|t)$ and apply $u^*(t) = \bar{u}^*(t|t)+K(x(t)-\bar{x}^*(t|t))$ to the system~\eqref{eq:model-linear}.

We assume that the terminal ingredients (i.e., $\ell_N$, $P$, $K$, and $\mathcal{X}_N$) satisfy the following assumption.
\begin{assumption}
There exist a matrix $P=P^\top>0$, a matrix $K \in \R^{n_u\times n_x}$, and a set $\mathcal{X}_N\subseteq \mathcal{X}\ominus \mathcal{S}$ such that, for all $\bar{x} \in \mathcal{X}_N$, the following conditions hold:
\begin{enumerate}
[label=(\roman*)]
    \item $\bar{A}_K\bar{x} \in \mathcal{X}_N$;  \label{ass:terminal-ing-i}
    \item $K\bar{x} \in \mathcal{U} \ominus K \mathcal{S}$;\label{ass:terminal-ing-ii}
    \item $\ell_N (\bar{A}_K\bar{x})-\ell_N (\bar{x})\leq -\ell(\bar{x},K\bar{x})$. \label{ass:terminal-ing-iii}
\end{enumerate}
\label{ass:terminal-ing}
\end{assumption}
Note that Assumption~\ref{ass:terminal-ing} is standard in MPC to ensure closed-loop stability \cite{rawlings2017model}. 

Before delving into the closed-loop guarantees of the proposed data-driven controller, we make some observations on the practical (offline) computation of the sets $\mathcal{S}$, $\mathcal{X}_N$, and the controller $K$.
\begin{remark}[Computation of $\mathcal{S}$]
To compute $\mathcal{S}$, we use the method proposed in \cite{rakovic2003approximation}, proven to be effective in practical applications. The approach is summarized in the following proposition.
\begin{proposition}
Consider the controlled error dynamics \eqref{eq:controlled-sys-RPI}. If there exist a $\kappa \in \Z_{\geq 0}$ and a $\theta\in [0,1)$ such that $\bar{A}_K^\kappa\zon_\phi \subseteq \theta \zon_\phi$, then 
\begin{equation*}
    \mathcal{S}(\theta,\kappa) = (1-\theta)^{-1} \bigoplus^{\kappa-1}_{i=0}\bar{A}_K^i\zon_\phi
\end{equation*}
is a convex, compact, and robustly positive invariant set of \eqref{eq:controlled-sys-RPI}. 
\end{proposition}
\end{remark}
\begin{remark}[Computation of $K$ and $P$]\label{remark:compute-K}
To compute $K$ and $P$ we follow \cite{Russo2023TZDDPC}. After over-approximating $\mzon_D$ by an interval matrix, we determine $K$ and $P$ to satisfy the inequality condition stated in Assumption~\ref{ass:localcontroller}, with $Q_K = Q + K^\top R K$ for its vertices. Due to practical considerations, including computational complexity associated with high dimensional systems, future research will focus on relaxing the hypothesis of a common Lyapunov function stated in Assumption~\ref{ass:localcontroller}.
\end{remark}

\begin{remark}[Computation of $\mathcal{X}_N$]
To compute $\mathcal{X}_N$, the approaches outlined in \cite{gruber2020computing,Berberich2021Terminal} can be employed. In this work, we define $\mathcal{X}_N = \{x\in \R^{n_x}| \norm{x-x_s}_P^2 \leq \alpha \}  \subseteq 
\{x\in \R^{n_x}|x \in \mathcal{X}\ominus \mathcal{S}, K x \in \mathcal{U}\ominus K\mathcal{S}\}$, where the matrix $P>0$ is obtained in Remark~\ref{remark:compute-K} and $\alpha>0$ is chosen such that the specified subset relationship is satisfied.
\end{remark}

\section{Closed-loop guarantees}\label{sec:guarantees}
We are now ready to prove the recursive feasibility and stability of the proposed approach. 
\begin{proposition}[p.~235 in \cite{rawlings2017model}]
\label{prop:cost-properties}
Let $\mathcal{F} \subseteq \mathcal{X}\ominus \mathcal{S}$ be the set of feasible solutions of \eqref{eq:opt-zonopc}.
Under Assumption~\ref{ass:terminal-ing}, $\forall \bar{x}^*(t) \in \mathcal{F}$ there exist constants $c_2>c_1>0$ such that for the cost $J^*_\ell(\bar{x}^*(t))$ the following conditions hold:
\begin{align}
    &J^*_\ell(\bar{x}^*(t)) \geq c_1  \norm{\bar{x}^*(t))}^2, \label{eq: J_lowerbound}\\
    &J^*_\ell(\bar{x}^*(t+1))-J^*_\ell(\bar{x}^*(t)) \leq -c_1  \norm{\bar{x}^*(t)}^2,\label{eq: J_difference}\\
    &J^*_\ell(\bar{x}^*(t)) \leq c_2  \norm{\bar{x}^*(t)}^2.\label{eq: J_upperbound} 
\end{align}
\end{proposition}

\begin{algorithm}[t]%\centering
\caption{Data-driven Tube-Based Zonotopic Predictive Control (TZPC)}
\label{alg:zonopc}
\textbf{Input} (from learning phase)\textbf{:} Prediction horizon $N$. 
Matrices $Q>0$, $R>0$, $P>0$, and $K$. Constraints $\mathcal{U}$, $\mathcal{X}$, and $\mathcal{X}_N$. Robustly positive invariant set $\mathcal{S}$. 
\begin{algorithmic}[1]
    \While{$t \in \Z_{\geq 0}$}
    \State\!\!\!\!\!Solve the optimal control problem~\eqref{eq:opt-zonopc}.
    \State\!\!\!\!\!Apply $u^*(t) = \bar{u}^*(t|t)+K(x(t)-\bar{x}^*(t|t))$ to the system~\eqref{eq:model-linear}.
    \State\!\!\!\!\!Increase the time step $t=t+1$.
   \EndWhile
\end{algorithmic}
\end{algorithm}

\begin{theorem}\label{thm:guarantees}
Suppose Assumptions~\ref{ass:localcontroller}--\ref{ass:terminal-ing} hold. Moreover, assume that the problem~\eqref{eq:opt-zonopc} is feasible at initial time $t = 0$. Then, the following closed-loop conditions are satisfied.
\begin{enumerate}[label=(\roman*)]
    \item The problem~\eqref{eq:opt-zonopc} is feasible at any $t \in \Z_{\geq 0}$; \label{thm:guarantees-i}
    \item The closed-loop trajectory satisfies the constraints, i.e., $x(t) \in \mathcal{X}$ and $u(t) \in \mathcal{U}$,  $ \forall t \in \Z_{\geq 0}$; \label{thm:guarantees-ii}
   \item The set $\mathcal{S}$ is robustly exponentially stable for the resulting closed-loop system. \label{thm:guarantees-iii} % which contains equilibrium $ [x_s \quad u_s]^\top$ 
\end{enumerate} 
\end{theorem}
\begin{proof}
\textit{\ref{thm:guarantees-i}--\ref{thm:guarantees-ii}.}
Without loss of generality, we let $u_s = \zero_{n_u}$ and $x_s = \zero_{n_x}$. We assume that the problem~\eqref{eq:opt-zonopc} is feasible at time $t$. At time $t + 1$, we construct a feasible candidate solution $(\bar{u}'_{t+1},\bar{x}'_{t+1})$ of \eqref{eq:opt-zonopc} by shifting the previously optimal solution  of~\eqref{eq:opt-zonopc} $(\bar{u}^*_t,\bar{x}^*_t)$ and appending $\bar{u}'(t+N|t+1) = K \bar{x}^*(t+N|t)$ with $K$ as in Assumption~\ref{ass:localcontroller}. Under Assumption~\ref{ass:terminal-ing}, this candidate solution satisfies all constraints of~\eqref{eq:opt-zonopc} and guarantees recursive feasibility. 

\textit{\ref{thm:guarantees-iii}.} We consider the cost of the optimal control problem \eqref{eq:opt-zonopc} $J^*_\ell(\bar{x}^*(t))$ as a candidate Lyapunov function.  From Proposition~\ref{prop:cost-properties}, we know that $J^*_\ell(\bar{x}^*(t))$ is lower- and upper-bounded. Moreover, according to~\eqref{eq: J_upperbound},
\begin{equation}
   -\norm{\bar{x}^*(t)}^2  \leq -\frac{1}{c_2 } J^*_\ell(\bar{x}^*(t)).
   \label{eq: J_upperbound-rewrite}
\end{equation}
Employing \eqref{eq: J_difference} together with \eqref{eq: J_upperbound-rewrite} leads to
\begin{align*}
    J^*_\ell(\bar{x}^*(t+1))-J^*_\ell(\bar{x}^*(t))
    %&\leq -c_1  \norm{\bar{x}^*(t)}^2\\
    %&\overset{\eqref{eq: J_upperbound-rewrite}}
    &\leq  -\frac{c_1}{c_2 } J^*_\ell(\bar{x}^*(t)),
\end{align*}
that is, 
%\begin{equation*}
$J^*_\ell(\bar{x}^*(t+1)) \leq (1-\frac{c_1}{c_2 })J^*_\ell(\bar{x}^*(t)).
$ %\end{equation*}
Define $\gamma = 1-\frac{c_1}{c_2 } \in (0,1)$. Then, we have 
\begin{equation*}
    J^*_\ell(\bar{x}^*(t)) \leq \gamma^t J^*_\ell(\bar{x}^*(0))\leq c_2\gamma^t\norm{\bar{x}^*(0)}^2,
\end{equation*}
which together with~\eqref{eq: J_lowerbound} leads to 
\begin{equation*}
  \norm{\bar{x}^*(t)} \leq \sqrt{\frac{c_2}{c_1 }} \sqrt{\gamma^t}\norm{\bar{x}^*(0)},
\end{equation*}
i.e., $\bar{x}^*(t)$ converges exponentially to the origin. Since  $x(t) \in \bar{x}^*(t) \oplus \mathcal{S}$, the set $\mathcal{S}$ is robustly exponentially stable for the controlled system~\eqref{eq:model-linear} \cite[Theorem~1]{mayne2005robust}.  
\end{proof}

\section{Numerical simulations}\label{sec:numerical-simulations}
This section presents two numerical examples. The first example is used to validate the performance of TZPC against TZDDPC \cite{Russo2023TZDDPC} and ZPC \cite{Alanwar2022Robust}. The second example illustrates the feasibility of TZPC in the context of building automation, where Assumption~\ref{ass:controllability} is not satisfied. Unless otherwise stated, for the simulations, we employed the toolbox CORA \cite{althoff2015introduction} in MATLAB, along with the Multi-Parametric Toolbox \cite{conf:mpt} and the YALMIP solver \cite{lofberg2004yalmip} \footnote{The MATLAB code for the corresponding examples are available in the GitHub repository at \href{https://github.com/MahsaFarjadnia/TZPC}{github.com/MahsaFarjadnia/TZPC}.}.
\begin{example}[Noisy double integrator]
\label{example:Alessio} This example is taken from \cite{Russo2023TZDDPC} and it represents a double integrator affected by bounded noise.
\begin{equation}
    x(k+1) =\begin{bmatrix} 1 & 1\\ 0 & 1\end{bmatrix} x(k)+ 
    \begin{bmatrix} 0.5\\ 1\end{bmatrix} u(k) +  w(k),
    \label{eq:example1}
\end{equation}
with $x(0)=[-5 \, -2]^\top$. The past data $D$ is constructed using 20 trajectories of length 5. The noise is assumed to be uniformly sampled from $\zon_w =  \left\langle \zero_{2},\begin{bmatrix} 0.02& 0.01\\0.01&0.02\end{bmatrix} \right\rangle$. The ingredients of the control problem are set as follows. The control input and state constraints are given in Fig.~\ref{fig:TZPC-ZPC_table}. The stage cost in~\eqref{eq:cost-function} is defined for any step $k$ as $\ell(x,u) = \norm{x}^2+ 10^{-2}|u|$. The prediction horizon is $N=7$. The terminal ingredients of \eqref{eq:opt-zonopc} are calculated as outlined in Section~\ref{subsec:Control} and are reported in Fig.~\ref{fig:TZPC-ZPC_table}. 

Fig.~\ref{fig:TZPC-ZPC} presents a comparative analysis between ZPC and TZPC. To ensure that both methods remain recursively feasible and are comparable, we have scaled the ZPC's state constraints by a 1.5 factor. Fig.~\ref{fig:reachablesets-ZPC-TZPC} illustrates the reachable sets of \eqref{eq:example1}, defined as the set of all states that the true system can reach at each $t$. These sets are given by 
$\mathcal{R}^x_{t+1} = \mzon_D ( \mathcal{R}^{x}_{t} \times \mathcal{R}^{u}_{t}) \oplus \zon_w,$ with $\mathcal{R}^{x}_{t} = \left\langle x(t), \zero_{n_x}\right\rangle$ and  $\mathcal{R}^{u}_{t} = \left\langle u^*(t), \zero_{n_u}\right\rangle$. 
Note that while ZPC uses reachability analysis, the reachable sets for TZPC are reported solely for illustrative purposes.

Additionally, we repeated the execution time analysis presented in \cite{Russo2023TZDDPC}, obtaining that the max. amount of time used by the solver (over 5 runs) is given by 40 min (ZPC), 78.26 min (TZDDPC, but only 10.12 min using its simplified version TZDDPC1, see \cite{Russo2023TZDDPC}), and 0.15 min (TZPC). To conclude, in addition to recursive feasibility and stability, the advantages of using TZPC include constraint satisfaction and shorter execution time. The latter is achieved by TZPC by circumventing the direct inclusion of reachable sets in the optimization constraints. Since the publicly available code for TZDDPC is implemented in Python, we computed the execution times for ZPC, TZPC, and TZDDPC using Python 3.8.18 in conjunction with the toolboxes mentioned in \cite{Russo2023TZDDPC} on a PC with an 11th Gen Intel\textsuperscript{\textregistered} Core\texttrademark~i7-1185G7 processor with 16.0 GB RAM. 
\end{example}

\begin{figure*}[!ht]\centering
\subfloat[]{\small\setlength\tabcolsep{0pt}
\begin{tabular}[b]{c}\hline
\rowcolor{gray!20}
Constraints\\\small\setlength\arraycolsep{2pt}
$\mathcal{U} =  \left\langle 0,1.3 \right\rangle$,\\
$\mathcal{X} = \left\langle \begin{bmatrix} -3.5\\ 0\end{bmatrix},\begin{bmatrix} 4 & 0\\ 0 & 2\end{bmatrix} \right\rangle$\\[15pt]\hline
\rowcolor{gray!20}
Terminal ingredients\\
$K = \begin{bmatrix} -0.107  &\! -0.603 \end{bmatrix}$\\[1pt]
$P = \begin{bmatrix} 0.895 & 0.492 \\ 0.492 & 3.709 \end{bmatrix}$, $\alpha=0.068$\\[10pt]
\setlength\arraycolsep{1pt}
$\mathcal{S} {=} \left\langle\begin{bmatrix} -0.038 \\ -0.036 \end{bmatrix}, \begin{bmatrix} 0.252 & 0 \\ 0 & 0.233 \end{bmatrix}\right\rangle$\\\hline
\end{tabular}\label{fig:TZPC-ZPC_table}
}\,
%--------
\subfloat[]{\includegraphics[height=4.4cm]{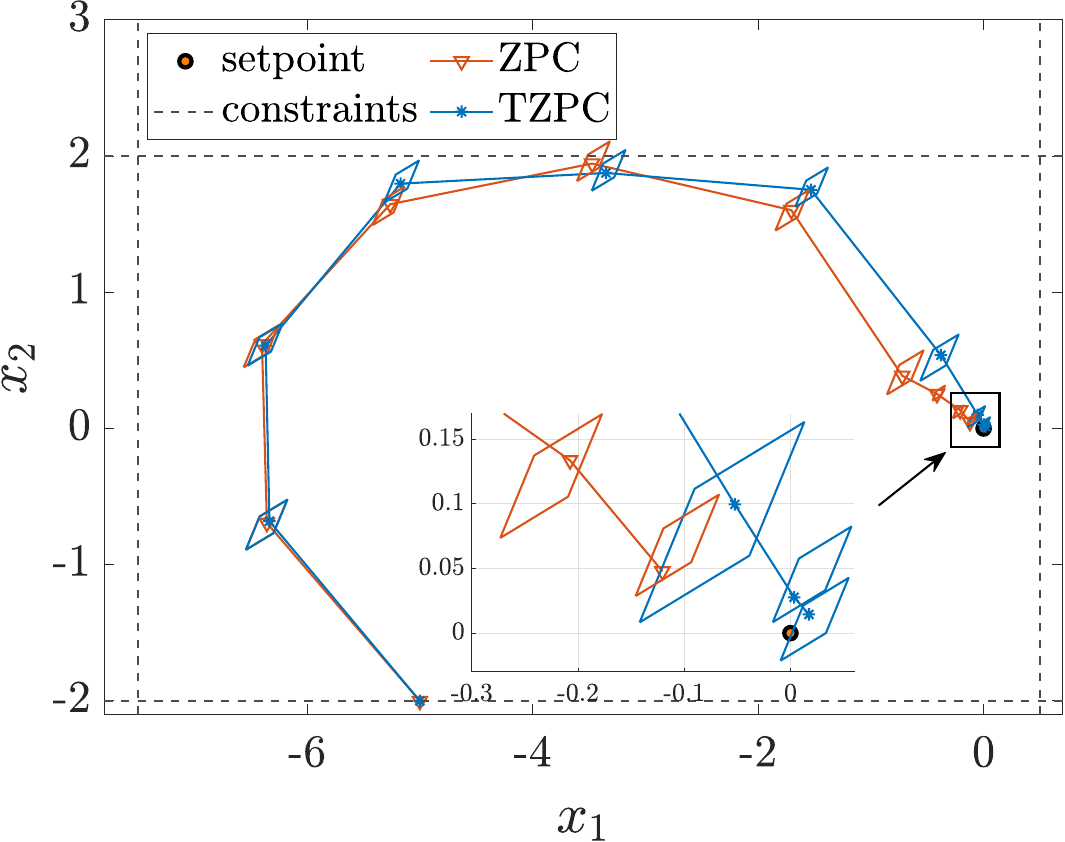}\label{fig:reachablesets-ZPC-TZPC}}\;
%--------
\subfloat[]{\includegraphics[height=4.4cm]{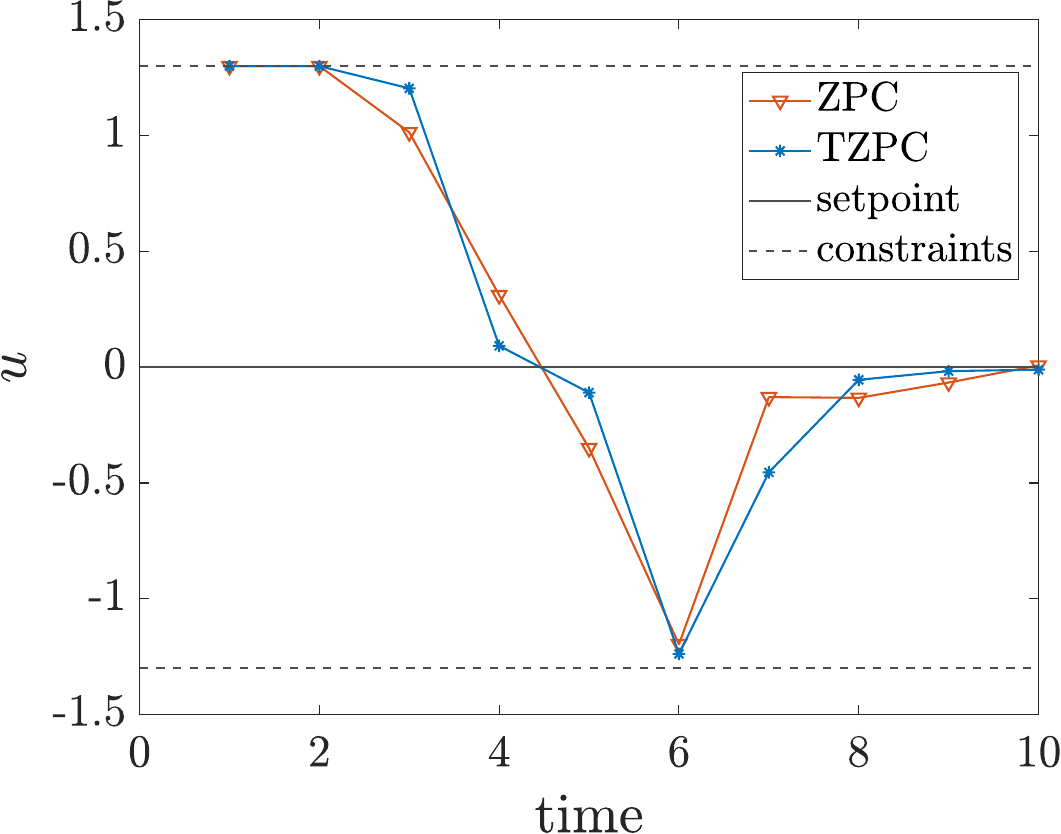}\label{fig:u-ZPC-TZPC}}
\caption{Example~\ref{example:Alessio}, comparison between ZPC and TZPC. (a): Ingredients of TZPC. (b): The reachable sets for the closed-loop system. (c): Optimal control input.}
\label{fig:TZPC-ZPC}
\end{figure*}

\begin{figure*}[!th]\centering
\subfloat[]{\small
\setlength\tabcolsep{0pt}
\begin{tabular}[b]{c}\hline
\rowcolor{gray!20}
Constraints\\\setlength\arraycolsep{2pt}
$\mathcal{U} = \left\langle 35,7 \right\rangle$\\ $\mathcal{X} = \left\langle \begin{bmatrix} 22\\ 20.5\end{bmatrix},\begin{bmatrix} 2 & 0\\ 0 & 1.25\end{bmatrix} \right\rangle$\\[15pt]\hline
\rowcolor{gray!20}Terminal ingredients\\
\setlength\arraycolsep{2pt}
%$K = \begin{bmatrix} 1.6039 &\! -5.2858 \end{bmatrix}$\\ 
$K = \begin{bmatrix} 1.604 &\! -5.286 \end{bmatrix}$\\[5pt] 
$P = \begin{bmatrix} 0.141 & 0.039 \\ 0.039 & 2.762 \end{bmatrix}$, 
%$P = \begin{bmatrix} 0.1408 & 0.0385 \\ 0.0385 & 2.7615 \end{bmatrix}$
%\\
$\alpha=0.0056$\\[10pt]
\setlength\arraycolsep{1pt}
$\mathcal{S} {=} 10^{-4}\left\langle\begin{bmatrix} -13 \\ -1 \end{bmatrix}, \begin{bmatrix} 3345 & 0 \\ 0 & 837 \end{bmatrix}\right\rangle$\\\hline
\end{tabular}\label{fig:building_table}
}\;
% ----------------
\subfloat[]{\includegraphics[height=4.7cm]{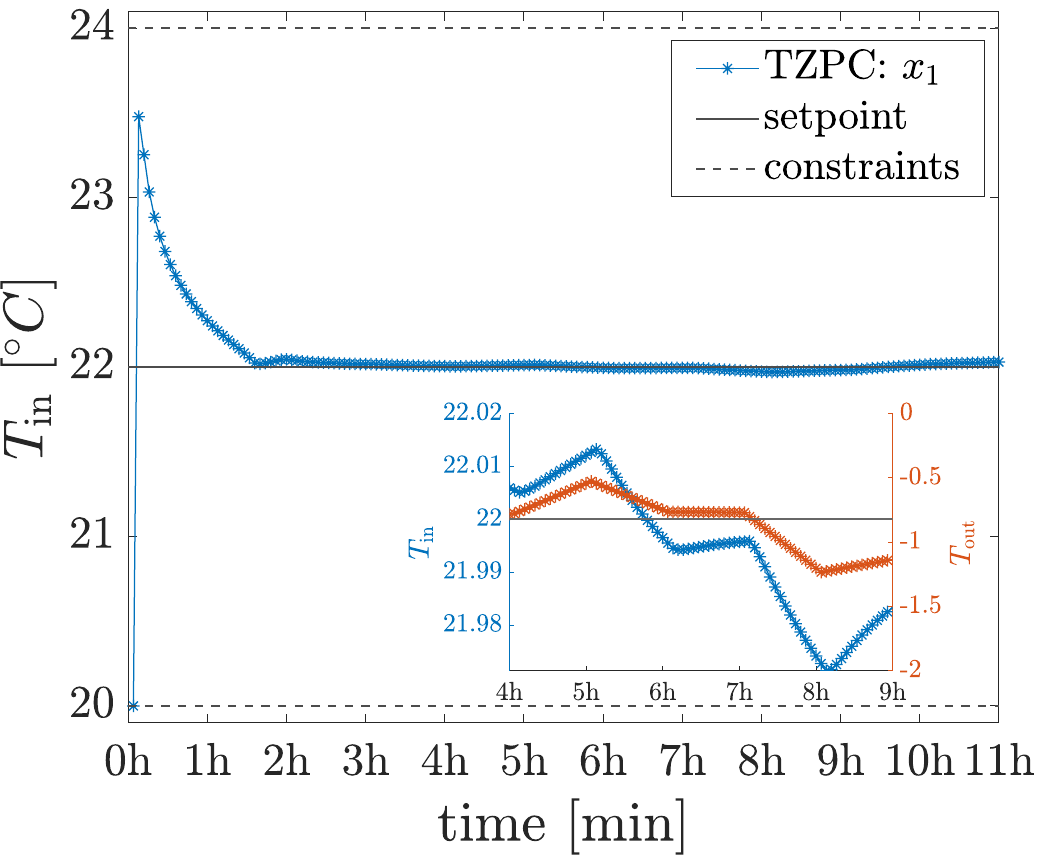}\label{fig:Tin}}\;
% ----------------
\subfloat[]{\includegraphics[height=4.6cm]{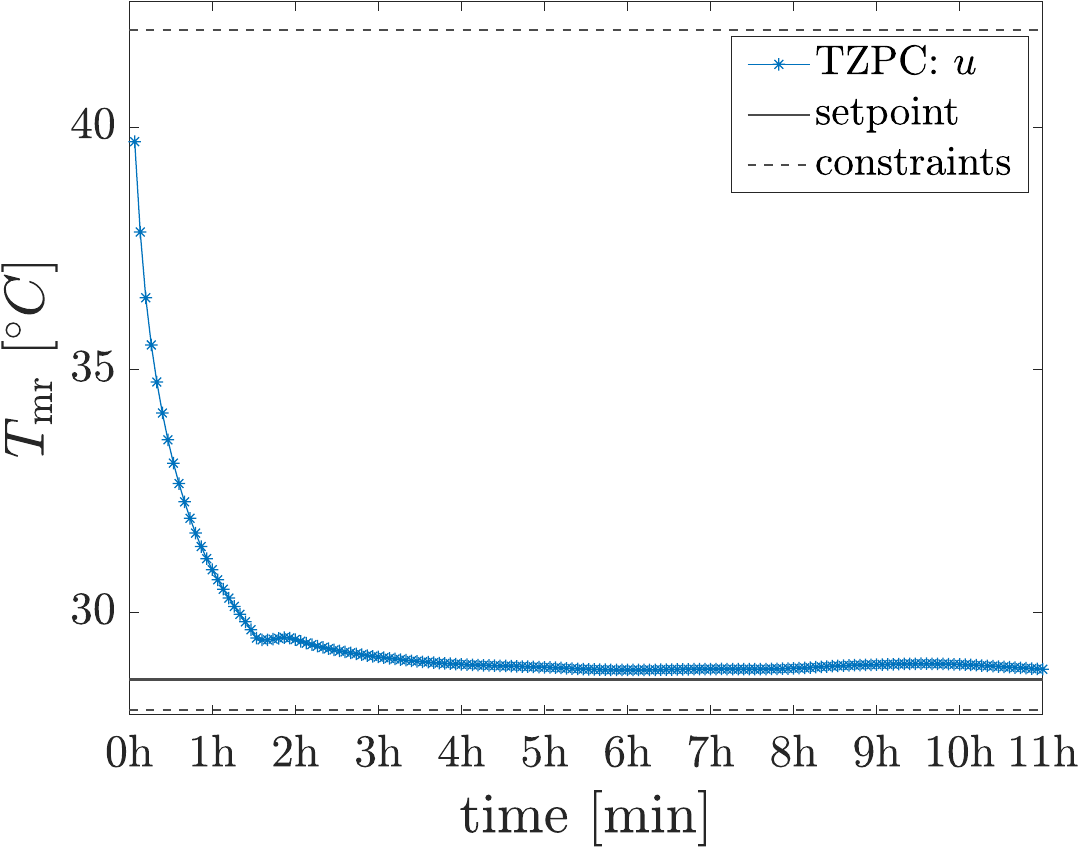}\label{fig:Tmr}}
\caption{Example~\ref{example:model-building}. (a): Ingredients of TZPC. (b): The closed-loop indoor temperatures trajectory. The inset illustrates the impact of the outside temperature on indoor conditions. (c): Optimal control input. The sample time is 4 min.}
\label{fig: building}
\end{figure*}

\begin{example}[Simplified building model]\label{example:model-building}
A review of methods to model the HVAC building system can be found, e.g., in  \cite{afroz2018modeling}. In this work, we are interested in thermal state-space models~\cite{Kulkarni2004Energy}, derived using energy balance equations. We utilize the following simplified thermal model of a single zone with radiators to showcase our algorithm in the context of residential buildings:
\begin{equation}
x(k+1) = \begin{bmatrix} 0.055  & 0.694\\ 0.043 & 0.956\end{bmatrix}  x(k) + \begin{bmatrix} 0.208\\ 0\end{bmatrix} u(k) + w(k).
\label{eq:model-building}
\end{equation}
In~\eqref{eq:model-building}, the state $x=[\Tin \; \Tw]^{\top}$ represents the indoor temperature ($\Tin$) and the temperature of the wall ($\Tw$),
the input $u=\Tmr$ is the mean radiant temperature of the radiators, and $w$ is the outside air temperature ($\Tout$).

The past data $D$ is constructed using 20 trajectories of length 5. We assume $\zon_w =  \left\langle 0,2 \right\rangle$. The setpoint temperature and control input are $x_s = \begin{bmatrix} 22 & 21.37\end{bmatrix}^\top$ and $u_s = 28.62$, respectively.
Note that, while the system is not controllable, it is stabilizable and the setpoint temperature is reachable. The ingredients of the control problem are set as follows.
The control input and state constraints, and the terminal ingredients are given in Fig.~\ref{fig:building_table}. The stage cost in \eqref{eq:cost-function} is  $\ell(x,u) = \norm{x-x_s}^2+ 10^{-2}|u-u_s|$, and $N = 11$.
The evolution of the indoor and outside temperatures and of the control input is depicted in Figs.~\ref{fig:Tin} and~\ref{fig:Tmr}, respectively. In this illustrative example, the impact of the disturbance is captured by the modest variation in $\Tout$, with an approximate magnitude of 0.5$^{\circ}C$ (see inset of Fig.~\ref{fig:Tin}). Hence, including other standard disturbances such as solar irradiation and internal heat gains from occupants, underscores the necessity for an effective robust controller within HVAC systems to ensure the maintenance of optimal indoor environmental conditions.

\end{example}

\section{Conclusion} \label{sec:conclusion}
This work introduces TZPC, a data-driven tube-based zonotopic predictive control approach designed for the robust control of unknown LTI systems using (noisy) input-state data. Through an offline data-driven learning phase and a subsequent online control phase, TZPC enables the formulation of a recursively feasible optimization problem, ensures the robust satisfaction of system constraints, and guarantees that the closed-loop system is robustly exponentially stable in the neighborhood of a desired setpoint. The effectiveness of the TZPC approach was demonstrated using two second-order examples. The first example enables a comparison with ZPC and TZDDPC, and the second one shows its applicability to stabilizable systems within the context of building automation.  Future work will showcase its applicability to higher-dimensional systems, which were not included in this study due to space constraints.

The theoretical guarantees provided by TZPC distinguish it from existing zonotopic-based approaches, such as ZPC \cite{Alanwar2022Robust} and TZDDPC \cite{Russo2023TZDDPC}. Additionally, this paper establishes the foundation for ensuring these closed-loop properties in the data-driven nonlinear zonotopic-based (NZPC) approach introduced in \cite{Farjadnia2023NZPC}. 

Future research will explore closed-loop guarantees in diverse cases, including time-varying constraints, input-output data, and bilinear and nonlinear unknown systems. Moreover, we will consider a comparative analysis with methods based on behavioral systems theory, such as those proposed in \cite{Berberich2021Guarantees}.

\bibliographystyle{IEEEtran}
\bibliography{setup,buildings,addrefs}

\end{document}